\documentclass{llncs}
\usepackage{amsmath,color}
\newcommand{\ie}{\textsl{i.e.}}
\newcommand{\BigOh}[1]{O\!\left(#1\right)}
\newcommand{\LittleOh}[1]{o\!\left(#1\right)}
\newcommand{\BigOmega}[1]{\Omega\!\left(#1\right)}
\newcommand{\LittleOmega}[1]{\omega\!\left(#1\right)}

\newcommand{\SU}[1]{\textsf{FF}_i\!\left(#1\right)}
\newcommand{\su}[1]{\textsf{FF}\!\left(#1\right)}
\DeclareMathOperator*{\argmin}{arg\,min}

\title{The Fresh-Finger Property \thanks{This research was partially supported by NSERC and MRI.}}
\author{John Howat\inst{1} \and John Iacono\inst{2} \and Pat Morin\inst{3}}
\institute{School of Computing\\Queen's University\\\email{howat@cs.queensu.ca}
\and Department of Computer Science and Engineering\\Polytechnic Institute of New York University\\\email{jiacono@poly.edu} \and
School of Computer Science\\Carleton University\\\email{morin@scs.carleton.ca}}

\begin{document}
\maketitle
\thispagestyle{plain}
\begin{abstract}
The \emph{unified property} roughly states that searching for an element is fast when the current access is close to a recent access. Here, \emph{close} refers to rank distance measured among all elements stored by the dictionary. We show that distance need not be measured this way: in fact, it is only necessary to consider a small working-set of elements to measure this rank distance. This results in a data structure with access time that is an improvement upon those offered by the unified property for many query sequences.
\end{abstract}


\section{Bounds for searching: notation and history}

Comparison-based searching is one of the most fundamental operations in computer science: given a set $S$ of $n$ totally ordered items, create a data structure that, given a query key $x$, will return the largest key in $S$ that is no larger than $x$. This is called \emph{predecessor search}. We focus on the case where $S$ is static, and thus can be assumed to be the integers from 1 to $n$. We refer to a search that returns $x$ as an \emph{access} to $x$. Let $A = \langle a_1, a_2, \ldots a_m \rangle$ denote a sequence of accesses to be performed on a data structure, with $m$ chosen to be sufficiently large to absorb any start-up costs.

Any comparison-based search data structure is, at its core, a method of choosing which comparisons to perform in order to execute an access. The data structure is essentially a way of encoding a comparison tree to execute each access---the data structure could do this in an explicit way like in a binary search tree, or in an implicit way as in the binary search algorithm. It has been long-known that information theory tells us that the worst-case time for an access must be $\BigOmega{\log n}$, and that $\BigOh{\log n}$ can be achieved with data structures such as binary search on a sorted array.

But, worst-case analysis is not the end of the story, as one can design data structures that execute operations in $\LittleOh{\log n}$ time if the operations have some kind of order to them. Thus, we can create data structures with access times that are functions of the access sequences themselves (or some distributional statistic of the access sequence)---these runtimes will still have $\BigOh{\log n}$ worst-case behavior, but will be faster over sequences that have certain desirable characteristics. We now review some runtime bounds that have been introduced, and the data structures whose runtimes have bounds (we will say a data structure has a bound to mean that its runtime can be bounded by the bound):

\paragraph{Static optimality bound.} If the number of searches in $A$ to $x$ is $f(x)$, then the runtime to search for $a_i$ is $\BigOh{\log \frac{m}{f(a_i)}}$.\footnote{In this paper,  $\log 0=\log 1=1$.} Knuth showed how to achieve this bound if the $f(\cdot)$-values are given in advance \cite{DBLP:journals/acta/Knuth71}.

\paragraph{Working-set bound.} Let $w_i(x)$ be the number of distinct items accessed since the last access to $x$ in $a_1, \ldots a_{i-1}$. A data structure has the working-set bound if an access to $a_i$ takes time $\BigOh{\log w_i(a_i)}$. The idea behind this is that if the accesses are restricted to a subset of $k$ items, then the accesses will take time $\BigOh{\log k}$ rather than $\BigOh{\log n}$. It has been shown that the working-set bound implies the static optimality bound in the amortized sense. Splay trees \cite{DBLP:journals/jacm/SleatorT85} have the working-set bound in the amortized sense, while the working-set structure \cite{DBLP:conf/soda/Iacono01a} was designed to have this bound in the worst-case.

\paragraph{Queueish bound \cite{DBLP:journals/algorithmica/IaconoL05}.} The working-set bound requires that items which were accessed recently take less time than those that have not been accessed in a while. The queueish bound reverses this and states that the time to access $a_i$ should be $\BigOh{\log( n-w_i(a_i))}$; thus any structure with the queueish bound will execute the least recently accessed item in constant time. No dictionary is known to have the queueish bound and it remains open whether such a dictionary can exist; however, it was shown that there is a structure with a close-to-queueish bound of $\BigOh{\log n-w_i(a_i)+\log \log n}$ amortized access time.

\paragraph{Dynamic finger bound \cite{DBLP:journals/siamcomp/ColeMSS00,DBLP:journals/siamcomp/Cole00}.} Let $d(x,y)$ be the number of keys between $x$ and $y$ in $S$ (this is just $|x-y|$ if $S$ is the integers from $1$ to $n$). The dynamic finger property says the cost to execute search $a_i$ is $\BigOh{\log d(a_{i-1},a_i)}$. Level-linked trees \cite{DBLP:journals/iandc/HoffmanMRT86} have the dynamic finger property in the worst-case, and splay trees have the amortized dynamic finger property.

\paragraph{Unified bound \cite{DBLP:journals/tcs/BadoiuCDI07}.} The dynamic finger bound and the working-set bounds are the best known bounds on the runtime of splay trees, yet neither implies the other and neither is tight. For example in the sequence $\langle 1,2,\ldots n, 1,2, \ldots \rangle$ the dynamic finger will give a bound of $\BigOh{1}$ per operation on average, while the working-set bound gives a bound of $\BigOh{\log n}$ per operation. In the sequence $\langle 1,n,1,n,\ldots \rangle$, the situation is reversed.  The unified bound was proposed as a natural combination of these two bounds;\footnote{In an unfortunate naming conflict, Sleator and Tarjan have a ``Unified Theorem'' for splay trees \cite[Theorem~5]{DBLP:journals/jacm/SleatorT85} and the bound in the Unified Theorem is also sometimes called the ``unified bound.''} 
an access is fast if it is close in key value to something that has been recently accessed. Formally, a data structure has the unified bound if accessing $a_i$ takes time $\BigOh{\log \min_j ( d(a_i,j)+w_i(j))}$. This clearly implies the working-set bound (set $j=a_i$) and the dynamic finger bound (set $j=a_{i-1}$). A non-tree structure was presented with the unified bound, and it is conjectured that splay trees have the unified bound. A binary search tree (BST) structure with the unified bound plus an additive $\BigOh{\log \log n}$ is known \cite{DBLP:conf/wads/DerryberryS09}, and a BST structure with the unified bound was claimed \cite{dthesis}, but later declared to be buggy \cite{wrong}.

\noindent
There are several issues that are important when considering a bound:

\paragraph{Static vs.~dynamic.} If a search algorithm uses the same search tree for every access, it is said to be \emph{static}, while if the comparisons performed to execute a given search depend upon the previous searches performed it is said to be \emph{dynamic}. The static optimality bound is the best bound possible if the search algorithm generates the same comparison tree for every access.

\paragraph{Online.} A bound where the runtime bound to execute $a_i$ is a function of the sequence $\langle a_1 \ldots a_i \rangle$ is said to be an online bound. All of the bounds listed above, except the static optimality, bound are online. The static optimality bound is not online because it is computed as a function of the frequency count over the entire length of a sequence.

\paragraph{Amortization.} For any operation there can be at most $\BigOh{2^k}$ different searches than can be done using at most $k$ comparisons. Any bound that at any time that has $\LittleOmega{2^k}$ different searches perform only $k$ comparisons for some value of $k$ means that the bound can not hold in the worst case. None of the bounds above require amortization, and if a bound does require amortization, that is probably a sign that it is somehow unnatural.

\paragraph{Binary Search Tree model.} Wilber \cite{DBLP:journals/siamcomp/Wilber89} formalized the binary search tree model; in this model the data structure is a binary search tree which can be restructured through the use of rotations. The set of sequences which binary search trees can execute quickly seems to be a reasonable classifier of those sequences that we consider to be natural. Without a restriction to the BST model, given any single access sequence, it is possible to create a data structure that will execute the searches in that sequence quickly, and others slowly. This is not possible in the BST model as there are deterministic sequences such as the bit reversal permutation that can not be executed faster than $\BigOh{\log n}$ amortized time.

The class of BST data structures also have the possibility that there may exist an online BST data structure that can execute every access sequence asymptotically as fast as the best BST data structure \emph{for that sequence}. Such a structure would be called \emph{dynamically optimal}; no BSTs are known to be dynamically optimal although spay trees and Lucas' trees \cite{lucas} are conjectured to be. Blum et.~al. \cite{blum} gave a non-tree data structure that runs within a constant factor of comparisons of any BST data structure, but requires superpoloynomial time to decide which comparisons to perform. Tango trees \cite{DBLP:journals/siamcomp/DemaineHIP07} are a BST that execute every sequence within a $\BigOh{\log \log n}$ factor of the best possible binary search tree. Of the bounds described above, no BST can have the queueish property, it is conjectured that there is a BST with the unified property, and the rest of the bounds described above are achievable by BST data structures.

\section{Problems with the unified bound}

The unified bound is the best proposed bound for binary search trees, and seems to be a reasonable combination of temporal locality and locality in keyspace. However, we will show that the unified bound has a flawed view of keyspace, and propose a new bound that attempts to rectify this flaw.

Recall that the unified property roughly states than access is fast when the current access is close to a recent access. For example, consider the following access sequence (assume $n$ is even and $n$ divides $m$): 
\begin{equation}
	A = \left\langle 1, \frac{n}{2}+1, 2, \frac{n}{2}+2, 3, \frac{n}{2}+3, \ldots, \frac{n}{2}, n \right\rangle^{m/n} \label{eq:bear}
\end{equation}
(The exponentiation denotes that the sequence is repeated $m/n$ times
to make a sequence of length $m$.)
Observe that, except for the first two accesses in each cycle, every access is at distance one from the element accessed two accesses ago.  A dictionary with the unified
property would therefore perform this sequences in time at most
\[
   2(m/n)\cdot \BigOh{\log n} + (n-2)(m/n)\cdot \BigOh{1} \in \BigOh{m}
\]
for an amortized cost of $\BigOh{1}$ per access.  

Next, consider the following access sequence:
\begin{equation}
	A' = \left\langle K, \frac{n}{2}+K, 2K, \frac{n}{2}+2K, 3K, \frac{n}{2}+3K, \ldots \frac{n}{2},n\right\rangle^{mK/2n}
\end{equation}
where $n$ is a multiple of $K$, $mK$ is a multiple of $n$, and $n^{1/4}\le K\le\sqrt{n}$. For this sequence, the unified bound is useless: any element accessed less than $K/2$ time units in the past is at distance at least $K$ from the currently accessed element, so the cost of every access is $\BigOmega{\log K} = \BigOmega{\log n}$.

On the other hand, the sequence $A'$ is not very different from $A$.  Indeed $A'$ can be viewed as the sequence $A$ over a larger set, $S$, in which a $(1-1/K)$ fraction of the elements are never accessed.  Intuitively, in a good data structure these irrelevant elements should ``fall out of the way'' in order to speed up accesses to the important elements (multiples of $K$).


The sequence $A'$ demonstrates the problem with the unified property: The distance function $d(x,y)$ simply measures the number of keys between $x$ and $y$. But, suppose some key values have not been accessed in a long time relative to $x$ and $y$, or in the extreme case, have never been accessed. Why should the number of such keys between $x$ and $y$ influence the runtime of accessing keys such as $x$ and $y$?
Put simply, they should not. Data structures such as splay trees will have items that are never accessed ``percolate'' to to bottom of the structure, and the runtime of a splay tree with a subtree of never-accessed keys is identical to the runtime if the keys are not there. Thus we need a more nuanced $d(\cdot,\cdot)$ function that ``forgets'' keys that have not been accessed in a while when computing key distance.

In this paper, we will expand on the idea of counting only recently accessed elements towards the distance between elements stored by the dictionary. The remainder of the paper is organized in the following way. In Section~\ref{section:definitions}, we define a new, stronger version of the unified property. In Section~\ref{section:main}, we show how to construct a dictionary that has this new property. We conclude with possible directions for future research in Section~\ref{section:conclusion}.

\section{Defining the fresh-finger property}
\label{section:definitions}

In this section, we define a new, stronger version of the unified property that we term the \emph{fresh-finger property}. Recall that $A = \langle a_1, a_2, \ldots, a_m \rangle$ is our access sequence. Define
\begin{displaymath}
	l_i(x) = \min \left( \{ \infty \} \cup \{ j > 0 \;|\; a_{i-j} = x \} \right)
\end{displaymath}

One can think of $l_i(x)$ as the most recent time $x$ has been queried in $A$ before time $i$. We then define
\begin{displaymath}
	w_i(x) =
		\begin{cases}
			n 																	& \text{if } l_i(x) = \infty \\
			|\{ a_{i-l_i(x)+1}, \ldots, a_i \}|	& \text{otherwise}
		\end{cases} 
\end{displaymath}
which is the \emph{working-set number} of $x$ at time $i$. We also define $W_i(j)$ to be the set of all elements $x \in S$ at time $i$ such that $w_i(x) \le j$, \ie, the set of all elements with working-set number at most $j$ at time $i$.  Next, we define $d_T(x,y)$ to be the rank distance between $a$ and $b$ in the set $T$, \ie,
\[
   d_T(x,y) 
    = \begin{cases}   
      |\{z\in T: x< z \le y\}| & \text{if $x<y$} \\
      |\{z\in T: y< z \le x\}| & \text{otherwise.}
    \end{cases}
\]
Finally, define
\begin{displaymath}
	y_i(x, T) = \argmin_{y \in T} w_i(y) + d_T(x,y)
\end{displaymath}

We are now ready to define the fresh-finger property. In terms of the preceding notation, the unified property states that the time to access the element $x \in S$ at time $i$ is 
\begin{equation}
	\BigOh{\log (w_i(y_i(x,S)) + d_S(x,y_i(x,S)))}
        \label{eq:u}
\end{equation}
A first attempt at defining the fresh-finger property might be 
\begin{equation}
	\BigOh{\log (w_i(y_i(x,W_i(w_i(x)))) + d_{W_i(w_i(x))}(x,y_i(x,W_i(w_i(x)))))}. \label{eq:su-1}
\end{equation}
Equation~\eqref{eq:su-1} should be contrasted with the definition of the unified property defined by \eqref{eq:u}.  In \eqref{eq:su-1}, the rank distance between $x$ and other elements is measured only with respect to the set $W_i(w_i(x))$, the set of elements that have been accessed since the last access to $x$ (\ie, a set of \emph{fresh fingers}).  In \eqref{eq:u}, the rank distance is measured with respect to $S$, the entire set of elements stored in the data structure.

Since $W_i(w_i(x))\subseteq S$, \eqref{eq:su-1} is certainly a stronger requirement.  Unfortunately, it is too strong, and there is no comparison-based data structure that can achieve this bound in the worst-case. To see this, consider the access sequence
\begin{displaymath}
	\left\langle 1,2,3,\ldots,n,1,x \right\rangle
\end{displaymath}
where $x \in \{1,\ldots,n\}$.  Let $i={n+2}$ (so that $a_i =x$ represents the second access to $x$).  Then $W_i(w_i(x)) = \{1,x,x+1,x+2,\ldots,n\}$.  But then the rank difference, $d_{W_i(w_i)}(x,1)$, between $x$ and $1$ is at most 1 and $w_i(1)=1$, so, according to \eqref{eq:su-1}, the time to accessing $x$ is at most\[
   \BigOh{\log(w_i(1)+d_{W_i(w_i)}(x,1)} = \BigOh{1}
\]
But this is true for any $x\in\{1,\ldots,n\}$, so for \emph{any} of the $n$ choices for $x$, \eqref{eq:su-1} requires that a data structure execute the access in constant time. This is not information-theoretically possible; accessing a randomly chosen $x\in\{1,\ldots,n\}$ requires at least $\log_2 n$ comparisons in expectation.  

From the preceding discussion, we conclude that the set in which we measure rank distance should be expanded.  This leads to the following definition:

\begin{definition} A data structure has the fresh-finger property if its runtime for an access is bounded by
\begin{displaymath}
	\BigOh{\log (w_i(y_i(x,W_i(w_i(x)^2))) + d_{W_i(w_i(x)^2)}(x,y_i(x,W_i(w_i(x)^2))))}.
\end{displaymath}
\end{definition}

Observe that the set $W_i(w_i(x)^2)$ contains elements that have been accessed less recently than $x$. These additional elements will allow us to support the rest of the access cost while respecting information-theoretic lower bounds. The intuition for expanding the set under consider in this manner is the fact that the data structure will consist of substructures that increase doubly-exponentially in size, and so by squaring the working-set number under consideration, we take advantage of elements in an adjacent substructure.

For brevity, we define
\begin{displaymath}
	y_i(x) = y_i(x,W_i(w_i(x)^2))
\end{displaymath}
and
\begin{displaymath}
	\SU{x} = \log (w_i(y_i(x)) + d_{W_i(w_i(x)^2)}(x,y_i(x)))
\end{displaymath}

As an example, consider the following access sequence, where $15$ is the element currently being accessed at the end of the sequence.
\begin{displaymath}
2, 3, 4, 5, 6, 7, 8, \overbrace{9, 10, 11, 12, 13, \underbrace{14, 15, 16, 1}_{W_i(w_i(15))}}^{W_i(w_i(15)^2)}, 15
\end{displaymath}

The original definition of the fresh-finger property uses $W_i(w_i(15))$, while the modified definition uses $W_i(w_i(15)^2)$. This modified definition allows for $9,10,11,12,13$ to contribute to the rank distance. This results in a query time that does not violate information-theoretic lower bounds, since it does not result in a situation where all $n$ queries must be executed in constant time. 
Before presenting a data structure that (nearly) achieves the fresh-finger property, it is worth doing a sanity-check of this definition.  In particular, we confirm that there exists (distributions over) sequences $A=\langle a_1,\ldots,a_m \rangle$ such that
\[
   \su{a_1,\ldots,a_m} = \sum_{i=1}^m \SU{a_i}
\]
is a lower-bound for accessing $a_1,\ldots,a_m$.

\begin{theorem}
  For all positive integers $n$, $r\le n$, and $m\ge 2r\log n$, there exists a distribution, $\mathcal{A}$, over $\{1,\ldots,n\}^m$ such that, for any comparison-based dictionary data structure, $D$, that stores $\{1,\ldots,n\}$, and an access sequence $A=a_1,\ldots,a_m$ drawn from $\mathcal{A}$
\begin{enumerate}
 \item $\su{a_1,\ldots,a_m}=\BigOh{m\log r}$.
 \item the expected number of comparisons performed by $D$ while accessing $A$ is $\BigOmega{m\log r}$.
\end{enumerate}
\end{theorem}

\begin{proof}
  The sequence $A$ is defined as $a_i=i$ for $i\le r$ or $a_i$ is selected
  uniformly at random from the set $\{1,\ldots,r\}$ for $r< i\le m$.  This
  choice of $A$ immediately implies that 
  \[
     (m-r)\log_2 r\ge (m/2)\log_2 r = \BigOmega{m\log r}
  \]
  is a lower-bound on the expected number of comparisons performed by $D$
  while accessing the (randomly chosen values) $a_{r+1},\ldots,a_m$.
  This establishes Part~2 of the result.

  On the the other hand, to establish Part~1, we have
  \[
      w_i(a_i) 
        \le \begin{cases}
          n & \text{for $i\le r$} \\
          r & \text{for $r< i\le m$}
        \end{cases}
  \]
  Thus,
  \[
     \su{a_1,\ldots,a_m} \le r\log n + (m-r)\log(2r) = \BigOh{m\log r}
  \]
  since $m\ge 2r\log n$.
\end{proof}

\section{Towards the fresh-finger property}
\label{section:main}

In this section, we describe a data structure that comes to within a small additive term of achieving the fresh-finger property.

\subsection{The data structure}
\label{section:main:datastructure}

The data structure consists of $k$ finger search trees $T_1,T_2,\ldots,T_k$ as well as $k$ accompanying queues $Q_1,Q_2,\ldots,Q_k$. Recall that finger search trees can support insertions and deletions in $\BigOh{1}$ worst-case time (when provided with a pointer to the element to be deleted) and finger searches in $\BigOh{\log d}$ worst-case time, where $d$ is the distance between the element being searched for and the supplied pointer into the data structure \cite{DBLP:journals/jcss/BrodalLMTT03}. 

The size of $T_j$ is $2^{2^j}$, except for $T_k$ which has size $n$. It follows that $k$ is $\BigOh{\log \log n}$. We will maintain the invariant that $T_j \subset T_{j+1}$ for all $1 \le j < n$. The queue $Q_j$ contains exactly the same elements as $T_j$ in the order they were inserted into $T_j$. Pointers are maintained between elements in the queue and corresponding elements in the finger search tree.

To perform a search, we will perform finger searches for $x$ in $T_1,T_2,\ldots$ until we find $x$ for the first time (say, $x \in T_j$). In $T_1$, we use an arbitrary element as the starting finger for the search. In all other trees, we run two finger searches for $x$ in parallel: one from the successor of the element found in the previous finger search tree, and one from the predecessor of the element found in the previous finger search tree. As soon as the first of these two searches terminates, we stop the other. 

To restructure the data structure after we have found $x \in T_j$, we must insert $x$ into $T_1,T_2,\ldots,T_{j-1}$ (note that $x$ is not present in any of these trees, since if it were, it would have already been found) and enqueue $x$ in $Q_1,Q_2,\ldots,Q_{j-1}$. At this point, we note that each of $T_1,T_2,\ldots,T_{j-1}$ and $Q_1,Q_2,\ldots,Q_{j-1}$ are too big. We therefore dequeue the oldest element in each of $Q_1,Q_2,\ldots,Q_{j-1}$ and delete the corresponding elements in $T_1,T_2,\ldots,T_{j-1}$.

\subsection{Analysis}
\label{section:main:analysis}

Recall that we are aiming for a running time of
\begin{displaymath}
	\BigOh{\SU{x}} = \BigOh{\log (w_i(y_i(x)) + d_{W_i(w_i(x)^2)}(x,y_i(x)))}
\end{displaymath}

Consider a search for $x$ at time $i$, and consider the element $y_i(x)$. Suppose $x$ first appears in $T_j$ and $y_i(x)$ first appears in $T_{j'}$. Because $x$ first appears in $T_j$, we have that $w_i(x) \ge 2^{2^{j-1}}$. Therefore, $j \le \log \log w_i(x) + \BigOh{1}$. Similar reasoning shows $w_i(y_i(x)) \ge 2^{2^{j'-1}}$, so that $j' \le \log \log w_i(y_i(x)) + \BigOh{1}$. We consider three cases, based on how $j$ compares with $j'$:

If $j \le j'$ (\ie, $x$ appears no later than $y_i(x)$),\footnote{In fact, this case can only occur when $j=j'$, since otherwise $w_i(y_i(x)) > w_i(x)$, and so $w_i(y_i(x)) + d(x,y) > w_i(x) + d(x,x)$, which contradicts the definition of $y_i(x)$.} then the running time follows easily: $x$ has working-set number $w_i(x) \le w_i(y_i(x))$. The element $x$ can thus be found in time $\sum_{l=1}^j 2^l = \BigOh{2^j}$, which is $\BigOh{2^{j'}} = \BigOh{\log w_i(y_i(x))}$.


The more interesting case occurs when $j > j'$ (\ie, $x$ appears after $y_i(x)$).  In this case, the algorithm will reach the tree, $T_{j'}$, containing $y_i(x)$ in time 
\[  
    \sum_{\ell=1}^{j'} \BigOh{\log 2^{2^\ell}}=\sum_{\ell=1}^{j'} 2^\ell = \BigOh{2^{j'}} = \BigOh{\log w_i(y_i(x))}
    \enspace .
\]
The search in $T_{j'}$ finds both the predecessor and and successor, $y_1$ and $y_2$ of $x$ in $T_{j'}$.  That is,
\[
       y_1 \le x \le y_2
\]
and $y_i(x)$ is not in the open interval $(y_1,y_2)$.  In particular,
for any set $T$ one of $y_1$ or $y_2$, say $y_1$, has
\[
    d_{T}(x,y_1) \le d_{T}(x,y_i(x))
\]
Indeed, from this point onwards, every search in $T_{\ell}$, for each $\ell\in\{j'+1,\ldots,j\}$, $y_\ell'$ such that
\[
    d_{T}(x,y_\ell') \le d_{T}(x,y_i(x)) \enspace .
\]

The elements in $T_{j'+1},\ldots,T_{j-1}$ are all in $W_i(w_i(x))$, and so
the remaining searches in $T_{j'+1},\ldots,T_{j-1}$ therefore take a total of at most
\begin{eqnarray*}
&&(j-j'-1)\BigOh{\log d_{W_i(w_i(x))}(x,y_i(x))}\\
&=&\BigOh{(\log d_{W_i(w_i(x))}(x,y_i(x)))(\log\log w_i(x))}
\end{eqnarray*}
time.


At last, the final search, in $T_j$ is the expensive one, since the only
guarantee we have on the elements of $T_j$ are that their working-set number
is at most $w_i(x)^2$.  Thus, the elements in $T_j$ are a subset of the elements in $W_i(w_i(x)^2)$ and the time to search in $T_j$ is at most
\[
    d_{W_i(w_i(x)^2)}(x,y_i(x)) \enspace .
\]

In either case, the total search time thus far is at most 
\begin{eqnarray*}
&&\BigOh{\log w_i(y_i(x))}\\&+& \BigOh{(\log d_{W_i(w_i(x))}(x,y_i(x)))(\log \log w_i(x)) + \log d_{W_i(w_i(x)^2)}(x,y_i(x))}
\end{eqnarray*}

At this point, $x$ has been found and we must now adjust the data structure. First, $x$ must be inserted in $T_1,T_2,\ldots,T_{j-1}$. Because we have a finger for $x$ inside each of these structures, this takes total time $\BigOh{\log \log w_i(x)}$. Enqueuing $x$ in each of $Q_1,Q_2,\ldots,Q_{j-1}$ also takes $\BigOh{j} = \BigOh{\log \log w_i(x)}$. The subsequent deletions and dequeueings of the oldest elements in $Q_1,Q_2,\ldots,Q_{j-1}$ and $T_1,T_2,\ldots,T_{j-1}$ take a total of $\BigOh{j} = \BigOh{\log \log w_i(x)}$ time as well, since the dequeueing operation takes $\BigOh{1}$ time and provides a pointer to the node in the corresponding tree where the deletion must be performed. Therefore, all restructuring operations take time $\BigOh{\log \log w_i(x)}$.

We therefore have

\begin{theorem}
\label{theorem:main}
There exists a static dictionary over the set $\{1,2,\ldots,n\}$ that supports querying element $x$ in worst-case time
\begin{displaymath}
\BigOh{\SU{x} + (\log d_{W_i(w_i(x))}(x,y_i(x)))(\log \log w_i(x))}
\end{displaymath}
\end{theorem}
	
\section{Conclusion}
\label{section:conclusion}

In this paper, we defined a stronger version of the unified property and described a data structure that achieves it to within a small additive term. Instead of computing rank distance over the entire dictionary, we compute rank distance only within a working-set containing an element that is close to a recently-accessed element.

There are several possible directions for future research.

\begin{enumerate}
\item One can measure distances within the set $W_i(w_i(x)^{1+\epsilon})$ instead of the set $W_i(w_i(x)^2)$ by changing how the substructures grow. Is it possible to reduce this further? For example, is it possible to measure distances within the set $W_i(\BigOh{w_i(x)})$?

\item We argued that it is not possible to measure within $W_i(w_i(x))$ in the \emph{worst case}. Is it possible to measure within this set in the amortized sense?

\item Can the additive term in Theorem~\ref{theorem:main} be reduced? It seems difficult to reduce this term below $\BigOmega{\log \log w_i(x)}$ using an approach similar to the one presented here, since elements must shift through at least this many substructures.

\item We have only considered the case where $S$ is static. Is it possible to maintain the fresh-finger property while supporting insertions into and deletions from $S$?
\end{enumerate}

\bibliographystyle{splncs}
\bibliography{dblp,other}

%
%
%
%
%
%
%
\end{document}